\documentclass[letterpaper, 10 pt, conference]{ieeeconf}
\usepackage[T1]{fontenc}
\IEEEoverridecommandlockouts                          

\usepackage{amsmath} 
\usepackage{amssymb} 
\usepackage{xcolor}
\usepackage{caption}
\usepackage{graphicx}
\usepackage{mathtools}
\usepackage{caption}
\usepackage{floatrow}
\usepackage{tikz}
\usepackage{tikz-cd}
\usepackage{url}

\usepackage{algorithmicx}
    \usepackage[ruled]{algorithm}
    \usepackage{algpseudocode}

\usepackage{pgfplots}
\usepackage{xcolor}
\usetikzlibrary{matrix,arrows,calc,positioning,shapes,decorations.pathreplacing}
\usepackage{graphicx}

\usepackage{amsmath} 

\usepackage{enumerate}
\usepackage[all,tips]{xy}
\SelectTips{cm}{11}

\newtheorem{theorem}{Theorem}

\newtheorem{definition}{Definition}
\newtheorem{remark}{Remark}
\newtheorem{proposition}{Proposition}

\usepackage{amssymb}

\newcommand{\proa}{A^*G \mbox{$\;$}_{\tau^*} \kern-3pt\times_\alpha
G \mbox{$\;$}_\beta \kern-3pt\times_{\tau^*} A^*G}

\begin{document}

\title{Stable Walking for Bipedal Locomotion under Foot-Slip\\ via Virtual Nonholonomic Constraints}

\author{ Leonardo Colombo$^{1}$, \'Alvaro Rodr\'iguez Abella$^{2}$, Alexandre Anahory Simoes$^{3}$, Anthony Bloch$^{4}$. 
\thanks{$^{1}$ (leonardo.colombo@car.upm-csic.es) is with Centre for Automation and Robotics (CSIC-UPM), Ctra. M300 Campo Real, Km 0,200, Arganda
del Rey - 28500 Madrid, Spain.}
\thanks{$^{2}$ (arabella@comillas.edu) is with the Department of Applied Mathematics, Comillas Pontifical University, Madrid, 28015 - Madrid, Spain.}
\thanks{$^{3}$ (alexandre.anahory@ie.edu) School of Science and Technology, IE University, Madrid, Spain.}
\thanks{$^{4}$ (abloch@umich.edu) is with the Department of Mathematics, University of Michigan, Ann Arbor, MI 48109, USA.}
\thanks{The authors acknowledge financial support from Grants  PID2022-137909NB-C21 funded by MCIN/AEI/10.13039/501100011033, and from  NSF grant  DMS-2103026 and AFOSR grants FA
9550-22-1-0215 and FA 9550-23-1-0400. A.R.A. acknowledges financial support from grant PID2024-156578NB-I00 funded by MICIU/AEI/10.13039/501100011033/FEDER, EU. L. C. was supported in part by iRoboCity2030-CM, Robótica Inteligente para Ciudades Sostenibles (TEC-2024/TEC-62), funded by the Programas de Actividades I+D en Tecnologías en la Comunidad de Madrid.
}}%

\maketitle
\begin{abstract}
Foot slip is a major source of instability in bipedal locomotion on low-friction or uncertain terrain. Standard control approaches typically assume no-slip contact and therefore degrade when slip occurs. We propose a control framework that explicitly incorporates slip into the locomotion model through virtual nonholonomic constraints, which regulate the tangential stance-foot velocity while remaining compatible with the virtual holonomic constraints used to generate the walking gait. The resulting closed-loop system is formulated as a hybrid dynamical system with continuous swing dynamics and discrete impact events. A nonlinear feedback law enforces both classes of constraints and yields a slip-compatible hybrid zero dynamics manifold for the reduced-order locomotion dynamics. Stability of periodic walking gaits is characterized through the associated Poincaré map, and numerical results illustrate stabilization under slip conditions.
\end{abstract}

\section{Introduction}

Robust locomotion on uncertain terrain remains a central challenge in legged robotics \cite{dosunmu2024demonstrating, verhagen2022human, siekmann2021blind}.
In particular, foot slip is a frequent source of performance degradation and loss of stability when walking on low-friction or poorly characterized surfaces \cite{kaneko2005slip, chen2016balance}.
Most control strategies for bipedal locomotion assume ideal no-slip contact conditions, as in the standard hybrid locomotion framework of \cite{Westervelt2007}, which simplifies the modeling and control design but leads to poor performance when tangential contact forces exceed the friction limits. As a result, slip is typically treated as a disturbance to be rejected rather than as a phenomenon that can be explicitly modeled and controlled \cite{kaneko2005slip}.

A large body of work has addressed gait stabilization using hybrid dynamical system models that combine continuous swing dynamics with discrete impact events at foot touchdown \cite{Westervelt2007, ames2012first}.
In this framework, periodic walking gaits are typically stabilized through the design of virtual constraints and the analysis of hybrid zero dynamics.
These techniques have been successfully applied to a wide range of robotic platforms and have led to strong theoretical guarantees such as orbital stability of periodic gaits \cite{Grizzle2018}.
However, the standard formulation assumes no-slip contact of the stance foot with the ground during the support phase, which restricts the applicability of the framework to high-friction environments. Modern optimization-based locomotion methods can explicitly enforce
contact and friction constraints; see, e.g., \cite{hereid2018dynamic}.
Our approach is complementary, incorporating regulated tangential
foot motion directly into the HZD framework.

When slip occurs, the stance foot is no longer fixed relative to the ground, and the resulting dynamics deviate from the nominal no-slip model. Rather than enforcing zero tangential velocity at the stance foot, it is natural to consider models in which tangential motion is allowed but regulated.
From a geometric perspective, this situation can be interpreted as the introduction of nonholonomic constraints associated with the allowable slip directions of the biped (for the general theory of nonholonomic constraints see
e.g., \cite{Bloch2003}).
Motivated by this observation, this paper proposes a control framework that explicitly incorporates slip into the locomotion dynamics through the design of \emph{virtual affine nonholonomic constraints}~\cite{stratoglou2023virtual}.
These constraints regulate tangential foot velocity and shape the closed-loop dynamics in a way that remains compatible with the hybrid structure of legged locomotion.

Virtual nonholonomic constraints are a class of virtual constraints that depend on the velocities of the system, rather than only on its configuration variables. They were introduced in~\cite{griffin2015nonholonomic,griffin2017nonholonomic} to design velocity-based swing-foot placement laws in bipedal robots. More recently, their geometric structure and control-theoretic role have been further developed in~\cite{simoes2023virtual,stratoglou2024geometry,simoes2025geometric}. In addition, this class of virtual constraints has been used in~\cite{horn2018hybrid,hamed2019nonholonomic,horn2020nonholonomic,horn2021nonholonomic} to encode velocity-dependent stable walking gaits through momenta conjugate to the unactuated degrees of freedom of legged robots and prosthetic legs.

Related slip-aware walking models based on hybrid zero dynamics enrich the reduced-order description by incorporating additional slip variables and slip-dependent gait transitions~\cite{chen2017hybrid,TrkovChenYi2019}. More recently, bipedal control under foot slip has also been addressed through integrated reduced-order and whole-body control \cite{mihalec2022integrated}. The viewpoint adopted here is different: instead of treating tangential foot motion only as an additional reduced-order effect, we regulate it directly at the level of the closed-loop kinematics through a virtual nonholonomic constraint designed to remain compatible with the nominal holonomic gait structure. In this way, admissible tangential motion is incorporated into the hybrid gait design through an explicit velocity-level mechanism.

The main contribution of this work is the construction of a slip-compatible hybrid zero dynamics manifold through a virtual affine nonholonomic constraint that regulates tangential stance-foot motion while remaining compatible with the nominal virtual holonomic gait constraints. We formulate a hybrid locomotion model in which admissible tangential contact motion is incorporated through a virtual affine nonholonomic constraint compatible with a nominal set of virtual holonomic constraints. We then derive a nonlinear feedback law that simultaneously enforces both classes of constraints and yields a slip-compatible hybrid zero dynamics manifold. Finally, we characterize local orbital stability of the resulting periodic gait through the associated Poincaré return map and validate the approach numerically under variable slip conditions.

The remainder of the paper is organized as follows.
Section~II introduces the locomotion model with foot-slip and reviews virtual holonomic constraints and the design of stable walking gaits used in locomotion. Section~III presents the design of the virtual nonholonomic constraints and develops the control design and stability analysis.
Section~IV illustrates the approach through numerical simulations.

\section{Hybrid Locomotion Framework}

\subsection{Hybrid locomotion model for walking with foot slip}\label{secIIA}

Bipedal walking is naturally modeled as a hybrid dynamical system consisting of continuous swing dynamics and discrete impact events. Let us briefly recall the fundamental facts of this model that will be useful in the forthcoming sections. For a full description of the bipedal walker, we refer for instance to~\cite{Westervelt2007}.


Let $Q\subset \mathbb{R}^n$ denote the configuration space of the free-floating model of the biped, with local coordinates $q \in Q$. The state space is thus the tangent bundle $TQ$, whose elements are denoted by $x=(q,\dot q)$ and represent positions and generalized velocities.
The continuous-time dynamics during the single-support phase is given by the Euler--Lagrange equations
\begin{equation}\label{ELeq}
D(q)\ddot q + C(q,\dot q)\dot q + G(q) = B(q) u + J_c(q)^\top \lambda ,
\end{equation}
where $D(q) \in \mathbb{R}^{n\times n}$ is the inertia matrix, $C(q,\dot q)\in\mathbb R^{n\times n}$ is the Coriolis matrix, including Coriolis and centrifugal terms, and $G(q)\in\mathbb R^n$ models gravitational forces.
The control input $u \in \mathbb{R}^m$ consists of the joint torques applied by the actuators, with 
$B(q) \in \mathbb{R}^{n\times m}$ denoting the actuation matrix that maps the control inputs to generalized forces.
Finally, $J_c(q)\in\mathbb R^{2\times n}$ denotes the Jacobian of the stance-foot contact constraint, and $\lambda=(\lambda_t,\lambda_n)\in\mathbb R^2$ contains the tangential and normal ground reaction forces.



Let $p_f(q) \in \mathbb{R}^2$ denote the Cartesian position of the stance foot, with velocity $v_f(q,\dot q)=\dot p_f(q,\dot q)=J_f(q)\dot q$,
where $J_f(q)=\frac{\partial p_f}{\partial q}(q)\in\mathbb R^{2\times n}$ is the Jacobian of the stance-foot contact point. In the present setting, this Jacobian coincides with the contact Jacobian introduced above, i.e., $J_f(q)=J_c(q)$. In the standard rigid-contact model for bipedal locomotion (see, for instance,~\cite[Eq. (3.19)]{Westervelt2007}), the stance foot is assumed to be fixed relative to the ground, that is, $v_f(q,\dot q)=0$. This no-slip assumption is valid only when the tangential contact force remains within the admissible friction cone, namely, $|\lambda_t|\leq \mu \lambda_n$, where $\mu>0$ is the static friction coefficient. When this friction bound is violated, the no-slip constraint breaks down and the stance foot may slide along the contact surface.

Throughout the paper, we restrict attention to a physically admissible
sliding-contact domain in which the normal contact force remains
compressive and the contact forces are compatible with the adopted
friction model. Denoting the corresponding admissible contact-force set
by $\mathcal{F}_{\rm slip}(q,\dot q)$, we require $\lambda=(\lambda_t,\lambda_n)\in
    \mathcal{F}_{\rm slip}(q,\dot q)$, $\lambda_n\geq 0$. Thus, the virtual constraint introduced below regulates the tangential
contact kinematics within this admissible domain; it does not replace
the physical contact conditions determining whether a prescribed slip
motion is realizable.

To account for slip, we relax the no-slip condition by allowing tangential motion at the contact while preserving normal contact.
Accordingly, the stance-foot velocity is decomposed as
$v_f = v_n + v_t$, where $v_n$ and $v_t$ are the normal and tangential components, respectively.
During the stance phase we impose the contact condition, $v_n = 0$, but allow for the tangential component to be non-zero, $v_t \neq 0$.
Thereby, tangential slip is explicitly incorporated into the locomotion model. 


When the swing foot strikes the ground, the hybrid system undergoes a discrete transition.
The switching surface is denoted by $S = \{(q,\dot q)\in TQ \mid h_{\mathrm{sw}}(q)=0,\;\dot h_{\mathrm{sw}}(q,\dot q)<0\}$. Note that if $p_{\mathrm{sw}}(q)=(p_{\mathrm{sw},x}(q),p_{\mathrm{sw},y}(q))\in\mathbb R^2$ denotes the Cartesian position of the swing foot, then $h_{\mathrm{sw}}(q)=p_{\mathrm{sw},y}(q)$ is its height relative to the ground. Thus, $S$ characterizes touchdown configurations for which the swing foot reaches the ground with negative normal velocity. At impact, we assume a perfectly inelastic collision, so that the configuration remains continuous while the velocity undergoes an instantaneous jump. The corresponding reset map is denoted by $\Delta:S\to TQ,\,\Delta(q,\dot q^-)=(q,\dot q^+)$, where $\dot q^-$ and $\dot q^+$ denote the pre- and post-impact velocities, respectively.
The post-impact velocity is determined by the impulse--momentum relation
\begin{equation}
D(q)(\dot q^+ - \dot q^-) = J_i(q)^\top \Lambda,
\end{equation}
where $J_i(q)\in\mathbb R^{2\times n}$ denotes the Jacobian of the impact contact constraint at touchdown and $\Lambda\in\mathbb R^2$ is the impulsive contact force.

Combining the continuous swing dynamics and the impact map, and denoting $x=(q,\dot q)\in TQ$, we obtain a hybrid system of the form
\begin{align}
\dot x &= f(x)+g(x)u, \qquad x\notin S,\label{cd}\\
x^+ &= \Delta(x^-), \qquad x^- \in S,\label{dd}
\end{align}
where $f$ and $g$ are induced by~\eqref{ELeq}. We refer to \eqref{cd}--\eqref{dd} as the hybrid locomotion system. System \eqref{cd}--\eqref{dd} has the standard hybrid structure used in bipedal locomotion, but differs from the classical nonslip form in that the continuous dynamics is written in free-floating coordinates and explicitly retains the contact forces, thereby allowing tangential slip during the stance phase.



\subsection{Virtual holonomic constraints and hybrid zero dynamics}

A common approach to stabilizing periodic walking gaits is based on the
design of virtual holonomic constraints (VHCs), which coordinate
the evolution of the robot joints through feedback control
\cite{Westervelt2007, ames2012first}.

A VHC is defined by an output function $y:Q \to \mathbb{R}^k$ that specifies desired relations among the configuration variables.
In locomotion, these relations are typically parameterized
through a phasing variable along the gait. Accordingly, suppose that a desired gait is encoded through the output
\begin{equation}\label{VHC}
y:Q\to\mathbb R^k,\quad y = h - h_d\circ\theta,
\end{equation}
where $h:Q\to\mathbb{R}^k$ provides the variables to be controlled and
$h_d\circ\theta:Q\to\mathbb R^k$ gives the desired evolution of such variables, which is parameterized by a phase
variable $\theta:Q\to\mathbb R$ that is assumed to be strictly
monotonic along the gait. In the standard VHC setting, one typically chooses the number of outputs to match the number of actuated inputs. In the present slip-aware setting, however, one scalar control direction is reserved for slip regulation, so that the holonomic output dimension is chosen as $k=m-1$.

The virtual constraint is enforced by driving the output to zero,
which defines the \emph{zero dynamics manifold}
\begin{equation}\label{Z}
\mathcal{Z} =
\{(q,\dot q) \in TQ \mid y(q)=0,~\dot y(q,\dot q)=0 \}.
\end{equation}
Since the output depends only on the configuration, its first time derivative is
$\dot y(q,\dot q)=\frac{\partial y}{\partial q}(q)\dot q$.
Therefore, $\mathcal Z$ is defined independently of the feedback law.

If the feedback controller renders $\mathcal{Z}$ invariant under the
continuous dynamics, the closed-loop system restricted to
$\mathcal{Z}$ defines the \emph{zero dynamics}. If $\mathcal{Z}$ is also invariant under the impact map,
i.e. $\Delta(\mathcal{Z}\cap S) \subset \mathcal{Z}$, then the dynamics restricted to $\mathcal{Z}$ define the
\emph{hybrid zero dynamics} (HZD).

The hybrid zero dynamics provide a reduced-order model that captures the gait evolution while satisfying the virtual constraints.
Stability of the walking gaits can then be analyzed by studying the periodic orbit of the hybrid system restricted to $\mathcal{Z}$; see \cite{westervelt2003hybrid} for details.

In the classical nonslip walking model, the stance foot is subject to a physical no-slip contact condition. When foot slip is present, this contact condition can no longer be assumed to hold, and the tangential stance-foot velocity must instead be regulated through feedback.
This motivates the introduction of virtual nonholonomic constraints, which encode the desired slip behavior while preserving the hybrid structure of the locomotion dynamics.



\subsection{Output dynamics, gait parameterization and stability}

The virtual constraints introduced above are enforced through outputs chosen to have relative degree two on the domain of interest. Consider the output function~\eqref{VHC} with state $x=(q,\dot q)\in TQ$. Differentiating twice along the dynamics~\eqref{cd} gives $\ddot y = L_f^2 y(x) + L_g L_f y(x)\,u$, where $L_f$ and $L_g$ denote Lie derivatives along the drift and input vector fields, respectively. Assume that the decoupling matrix $L_gL_f y(x)$ is nonsingular on the domain of interest. Then the input--output linearizing controller
\[
u=-(L_gL_f y)^{-1}\big(L_f^2 y+K_dL_f y+K_p y\big),
\]
with $K_p,K_d\in\mathbb R^{k\times k}$ positive definite, yields the closed-loop output dynamics $\ddot y + K_d \dot y + K_p y = 0$. Consequently, the output converges exponentially to zero, and the zero dynamics manifold~\eqref{Z} is rendered invariant and locally exponentially attractive.

The desired gait is typically parameterized using Bézier polynomials. Let $\displaystyle{h_d(\theta) =
\sum_{i=0}^{M} \alpha_i
\binom{M}{i}
\theta^i (1-\theta)^{M-i}}$, where $\theta(q)$ is a strictly monotonic phasing variable that
parameterizes the progression of the gait along the step, and
$\alpha_i$ are Bézier coefficients specifying the desired gait profile.
The coefficients at the endpoints determine the desired values at the beginning and end of the step, while the remaining coefficients shape the transition. 


The stability of the resulting periodic walking gait is analyzed
through the Poincaré map associated with the hybrid system \cite{morris2005restricted, goodman2019existence}.
Let $P$ denote the Poincaré return map defined on a transversal
section of the switching surface $S$.
If $x^\ast$ is a fixed point of $P$, i.e., $P(x^\ast) = x^\ast$, then the corresponding orbit represents a periodic walking gait. Orbital stability of the gait is determined by the eigenvalues of the
linearized Poincaré map $A_P=DP(x^\ast)$. If all eigenvalues lie strictly inside the unit circle, the periodic
gait is locally exponentially stable.

\section{Virtual Nonholonomic Constraints for Slip Regulation and Stable Walking}

Next, we introduce virtual affine nonholonomic constraints that regulate the slip velocity through feedback control and make them compatible with the virtual holonomic constraints used to encode the nominal walking gait.

\subsection{Virtual affine nonholonomic constraints}

To explicitly account for stance-foot slip, we introduce a class of velocity constraints that prescribe admissible tangential motion at the contact point. Let $p_f(q)=(p_{f,x}(q),p_{f,y}(q))\in\mathbb R^2$ denote the Cartesian position of the stance foot. In the planar setting considered here, the contact surface is locally modeled as the horizontal ground $\mathcal C=\{(x,y)\in\mathbb R^2 \mid y=0\}$, so that its tangent space is spanned by the horizontal direction $e_t=(1,0)^\top$. Hence, if $J_f(q)=\frac{\partial p_f}{\partial q}(q)\in\mathbb R^{2\times n}$, the tangential stance-foot velocity is given by $v_t(q,\dot q)=e_t^\top J_f(q)\dot q$.

Instead of imposing the no-slip condition $v_t=0$, we prescribe an admissible slip law of the form
\begin{equation}
v_t(q,\dot q)=b(q),
\label{eq:slip_law_sec3}
\end{equation}
where $b:Q\to\mathbb{R}$ is a smooth map specifying the desired tangential slip profile. Equation~\eqref{eq:slip_law_sec3} can be written as
\begin{equation}
A(q)\dot q=b(q),
\label{eq:affine_constraint_sec3}
\end{equation}
with $A(q):=e_t^\top J_f(q)\in\mathbb R^{1\times n}$. This defines, for each $q\in Q$, an affine subspace of admissible velocities,
\begin{equation}
\mathcal D(q)=\{\,\dot q\in T_qQ \mid A(q)\dot q=b(q)\,\}.
\label{eq:affine_distribution_sec3}
\end{equation}

The admissible velocities in \eqref{eq:affine_distribution_sec3} are understood together with
the sliding-contact admissibility conditions introduced in
Section~\ref{secIIA}. Hence, the virtual affine nonholonomic constraint specifies
the desired tangential kinematics, while contact-force feasibility is a
separate physical requirement on the resulting closed-loop trajectory.

\begin{definition}
The affine family of admissible velocities \eqref{eq:affine_distribution_sec3} is said to define a \emph{virtual affine nonholonomic constraint} if there exists a feedback law $u=\alpha(q,\dot q)$ such that the set
\[
\mathcal Z_{\mathrm{slip}}
:=
\{(q,\dot q)\in TQ \mid A(q)\dot q=b(q)\}
\]
is invariant under the closed-loop continuous-time dynamics.
\end{definition}

In other words, if the initial condition satisfies the affine velocity constraint, then the closed-loop trajectory satisfies it for all future times during the stance phase. In this case, the slip motion is not treated as a disturbance, but rather as a regulated component of the locomotion dynamics.

Whereas virtual holonomic constraints coordinate configuration variables through outputs of relative degree two, the constraints in \eqref{eq:affine_constraint_sec3} act directly at the velocity level and encode admissible tangential contact motion during stance.

Define the slip output $\eta_s:TQ\to\mathbb{R}$ by $\eta_s(q,\dot q):=A(q)\dot q-b(q)$. Along the continuous-time dynamics \eqref{ELeq} its derivative is
\begin{equation}
\dot\eta_s
=
\Gamma_s(q,\dot q,\lambda)+M_s(q)u,
\label{eq:slip_output_dynamics}
\end{equation}
where
\begin{align}
\Gamma_s(q,\dot q,\lambda)
&:=
A(q)D(q)^{-1}\!\left(-C(q,\dot q)\dot q-G(q)+J_c(q)^\top\lambda\right)
\nonumber\\
&\qquad
+\dot A(q,\dot q)\dot q-\frac{\partial b}{\partial q}(q)\dot q,
\label{eq:Gamma_s}
\\
M_s(q)&:=A(q)D(q)^{-1}B(q)\in\mathbb R^{1\times m}.
\label{eq:Ms}
\end{align} 
{We assume that the contact force $\lambda$ is available to the
controller, either from force sensing or from a contact-force observer.
Accordingly, $\lambda$ is treated as a measured/estimated signal in
\eqref{eq:slip_output_dynamics}-\eqref{eq:Gamma_s}, rather than as an
unknown algebraic variable to be determined by the feedback law. Here, the stance phase domain denotes the subset of $TQ$ in which the stance foot remains in contact with the ground and no impact event occurs, so that the continuous-time dynamics \eqref{ELeq} are valid.

\begin{proposition}
Assume that $A(\cdot)$ has constant rank $1$ on the stance phase domain, and that $M_s(q)$ in \eqref{eq:Ms} has full row rank $1$ for every $q$ in that domain. Then, locally on the stance phase domain, there exists a smooth feedback law rendering the constraint \eqref{eq:affine_constraint_sec3} invariant under the closed-loop continuous-time dynamics.
\end{proposition}

\begin{proof}
Consider the constraint function $F:TQ\to\mathbb{R}$, $F(q,\dot q):=A(q)\dot q-b(q)=\eta_s(q,\dot q)$. By assumption, $A(q)$ has constant rank $1$; hence $F$ has constant rank $1$ with respect to the fiber variable $\dot q$. Therefore, $\mathcal Z_{\mathrm{slip}}=F^{-1}(0)$ is an embedded submanifold of $TQ$ of codimension $1$.

Differentiating $F$ along trajectories of \eqref{ELeq} yields \eqref{eq:slip_output_dynamics}. Since $M_s(q)$ has full row rank, for each $q$ there exists a smooth local right inverse $M_s^\dagger(q)$ satisfying
$M_s(q)M_s^\dagger(q)=1$. Fix any scalar gain $k_s>0$ and define
\begin{align}
u
=&
M_s^\dagger(q)\bigl(-\Gamma_s(q,\dot q,\lambda)-k_s\eta_s(q,\dot q)\bigr).
\label{eq:slip_invariance_control}
\end{align}
Substituting \eqref{eq:slip_invariance_control} into \eqref{eq:slip_output_dynamics} gives
\[
\dot\eta_s
=
\Gamma_s+M_sM_s^\dagger(-\Gamma_s-k_s\eta_s)
=
-k_s\eta_s.
\]
Hence $\eta_s(\cdot)$ satisfies the linear homogeneous equation $\dot\eta_s=-k_s\eta_s$. If the initial condition lies in $\mathcal Z_{\mathrm{slip}}$, then $\eta_s(0)=0$ and therefore $\eta_s(t)\equiv 0$ for all times for which the solution exists. Thus $\mathcal Z_{\mathrm{slip}}$ is invariant under the closed-loop continuous-time dynamics.
\end{proof}

\begin{remark}
Since the slip constraint is scalar in the present planar model, the full-row-rank condition on $M_s(q)$ reduces to $M_s(q)\neq 0$. This is the coordinate-level algebraic counterpart of the transversality condition between the affine constraint distribution and the input distribution in \cite{stratoglou2023virtual}.\hfill$\diamond$
\end{remark}

\begin{remark}
If $\mathcal Z_{\mathrm{slip}}$ is rendered invariant by feedback, then the closed-loop continuous-time dynamics restrict to a smooth reduced-order dynamics on $\mathcal Z_{\mathrm{slip}}$. This restricted dynamics describes the evolution of the mechanical system under the admissible slip kinematics encoded by the affine constraint $A(q)\dot q=b(q)$.\hfill$\diamond$
\end{remark}
\subsection{Control Design}

We now construct a feedback law that enforces the virtual affine nonholonomic constraint while simultaneously stabilizing the virtual holonomic outputs defining the desired gait.

In addition to regulating slip, the controller must enforce the virtual holonomic constraints defining the desired walking gait. Using the relative-degree-two structure described in Section~II, the output dynamics can be written as
\begin{equation}
\ddot y=H(q,\dot q,\lambda)+M_h(q,\dot q)u,
\label{eq:y_ddot_sec3}
\end{equation}
where $H(q,\dot q,\lambda):=L_f^2y(x)$ and $M_h(q,\dot q):=L_gL_fy(x)$.

Stacking \eqref{eq:slip_output_dynamics} and \eqref{eq:y_ddot_sec3}, we obtain the mixed relative-degree system
\begin{equation}
\begin{bmatrix}
\dot\eta_s\\[1mm]
\ddot y
\end{bmatrix}
=
\begin{bmatrix}
\Gamma_s\\[1mm]
H
\end{bmatrix}
+
\underbrace{\begin{bmatrix}
M_s\\[1mm]
M_h
\end{bmatrix}}_{=: \,\mathcal A(q,\dot q)}
u.
\label{eq:stacked_output_dynamics}
\end{equation}

Since $M_h(q,\dot q)\in\mathbb{R}^{(m-1)\times m}$, assume that $\operatorname{rank}M_h=m-1$ and let $n_h(q,\dot q)$ span $\ker M_h(q,\dot q)$. Then $\mathcal A(q,\dot q)$ is nonsingular if and only if $M_s(q)n_h(q,\dot q)\neq 0$, i.e., the slip-control direction is transverse to the null direction left by the holonomic outputs.

Assume that the contact force $\lambda$ is available as described above and that the closed-loop evolution remains in a regularity domain where $\sigma_{\min}(\mathcal A(q,\dot q))\geq\varepsilon>0$, here $\sigma_{\min}$ is the smallest eigenvalue of $\mathcal A(q,\dot q)$. This also provides a robustness margin, since nonsingularity is preserved under perturbations $\Delta\mathcal A$ satisfying $\|\Delta\mathcal A\|_2<\varepsilon$, where $m=k+1$ is the number of control inputs and $\|\cdot\|_{2}$ is the Frobenius norm. Let $k_s>0$ and let $K_p,K_d\in\mathbb{R}^{k\times k}$ be positive definite gain matrices. Consider the feedback law
\begin{equation}
u=
\mathcal A(q,\dot q)^{-1}
\begin{bmatrix}
-\Gamma_s(q,\dot q,\lambda)-k_s\eta_s\\[1mm]
-H(q,\dot q,\lambda)-K_d\dot y-K_p y
\end{bmatrix}.
\label{eq:combined_feedback}
\end{equation}
Under \eqref{eq:combined_feedback}, the closed-loop output dynamics become
\begin{align}
\dot\eta_s+k_s\eta_s&=0,
\label{eq:closed_loop_eta_s}
\\
\ddot y+K_d\dot y+K_p y&=0.
\label{eq:closed_loop_y}
\end{align}

Substituting the control law into the system dynamics yields the closed-loop hybrid system
\begin{align}
\dot x &= f_c(x), \quad x\notin S, \label{eq:closed_loop_flow_sec3}\\
x^+ &= \Delta(x^-), \quad x\in S. \label{eq:closed_loop_jump_sec3}
\end{align}

Hence, the manifold defined by $\eta_s=0$, $y=0$, and $\dot y=0$ is invariant under the closed-loop continuous-time dynamics, and the corresponding outputs converge exponentially to zero. The impact map governs the discrete transitions between successive steps.

\begin{proposition}
Assume that $\mathcal A(q,\dot q)$ is nonsingular on the stance phase domain, and define
\begin{align}
&\mathcal Z_{\mathrm{int}}
=
\mathcal Z\cap \mathcal Z_{\mathrm{slip}}
\label{zint}\\
&=
\left\{(q,\dot q)\in TQ \;\middle|\; y(q)=0,\ \frac{\partial y}{\partial q}(q)\dot q=0,\ A(q)\dot q=b(q)\right\}.
\nonumber
\end{align}
Assume moreover that the impact map satisfies
\begin{equation}
\Delta\bigl(\mathcal Z_{\mathrm{int}}\cap S\bigr)\subset \mathcal Z_{\mathrm{int}}.
\label{eq:impact_compatibility_zint}
\end{equation}
Then, under the feedback law \eqref{eq:combined_feedback}, the manifold $\mathcal Z_{\mathrm{int}}$ is invariant under the closed-loop continuous-time dynamics, and for every $x^-\in \mathcal Z_{\mathrm{int}}\cap S$ the corresponding post-impact state satisfies $x^+=\Delta(x^-)\in \mathcal Z_{\mathrm{int}}$. In particular, $\mathcal Z_{\mathrm{int}}$ defines a hybrid zero dynamics manifold compatible with slip.
\end{proposition}

\begin{proof}
Let $x=(q,\dot q)$ and suppose that $x(0)\in\mathcal Z_{\mathrm{int}}$. By definition of $\mathcal Z_{\mathrm{int}}$, we have
\[
y(0)=0,\, \frac{\partial y}{\partial q}(q(0))\dot q(0)=0,\, A(q(0))\dot q(0)-b(q(0))=0.
\]
Equivalently, $y(0)=0,\, \dot y(0)=0,\, \eta_s(0)=0$.

Under the feedback law \eqref{eq:combined_feedback}, the slip output satisfies the closed-loop equation
$\dot\eta_s+k_s\eta_s=0$. Since this is a linear homogeneous differential equation with initial condition $\eta_s(0)=0$, uniqueness of solutions implies that $\eta_s(t)\equiv 0$ for all times for which the solution exists. Hence the affine slip constraint is preserved along the continuous-time dynamics.

Likewise, the gait output satisfies $\ddot y+K_d\dot y+K_p y=0$. Since this is a linear homogeneous second-order system with initial conditions $y(0)=0$ and $\dot y(0)=0$, uniqueness of solutions implies that $y(t)\equiv 0,\, \dot y(t)\equiv 0$ for all times for which the solution exists. Hence the virtual holonomic constraints are preserved along the continuous-time dynamics.

Therefore, every trajectory initialized in $\mathcal Z_{\mathrm{int}}$ remains in $\mathcal Z_{\mathrm{int}}$ under the closed-loop continuous-time flow. That is, $\mathcal Z_{\mathrm{int}}$ is invariant under the flow \eqref{eq:closed_loop_flow_sec3}.

Now let $x^-\in \mathcal Z_{\mathrm{int}}\cap S$. By \eqref{eq:impact_compatibility_zint} the corresponding post-impact state satisfies $x^+=\Delta(x^-)\in \mathcal Z_{\mathrm{int}}$. Hence $\mathcal Z_{\mathrm{int}}$ is preserved by the impact map whenever the trajectory reaches the switching surface. It follows that $\mathcal Z_{\mathrm{int}}$ is invariant under both the continuous-time closed-loop dynamics and the discrete reset map. 

Therefore, $\mathcal Z_{\mathrm{int}}$ defines a hybrid invariant manifold on which both the virtual holonomic constraints and the slip constraint are enforced. In particular, it defines a hybrid zero dynamics manifold compatible with slip.
\end{proof}

\begin{definition}
The manifold $\mathcal Z_{\mathrm{int}}$ is called a \emph{slip-compatible hybrid zero dynamics manifold}.
\end{definition}

\begin{remark}
If the feedback law \eqref{eq:combined_feedback} is well defined and $\mathcal Z_{\mathrm{int}}$ is invariant under the continuous-time closed-loop dynamics, then the vector field $f_c$ restricts to a well-defined reduced-order continuous-time dynamics on $\mathcal Z_{\mathrm{int}}$. Along this restricted dynamics, the virtual holonomic constraints are exactly enforced and the tangential stance-foot velocity satisfies the prescribed slip law. Thus, $\mathcal Z_{\mathrm{int}}$ provides a reduced-order stance model that simultaneously captures walking dynamics and regulated slip motion.\hfill$\diamond$
\end{remark}

\subsection{Stability of Periodic Walking Gaits}

Let $x=(q,\dot q)$ denote the state of the closed-loop hybrid system and consider the dynamics \eqref{eq:closed_loop_flow_sec3}--\eqref{eq:closed_loop_jump_sec3}. A walking gait corresponds to a periodic orbit of the hybrid system. More precisely, a trajectory $x^\ast(t)$ represents a periodic gait if there exists $T>0$ such that $x^\ast(T^+) = x^\ast(0)$.

Under the virtual holonomic constraints introduced in Section~II and the virtual affine nonholonomic slip constraint developed in this section, the manifold $\mathcal Z_{\mathrm{int}}$ is invariant under the closed-loop continuous-time dynamics. To study stability of a periodic gait, we introduce coordinates measuring deviation from $\mathcal Z_{\mathrm{int}}$. Let $\eta =
\begin{bmatrix}
y\quad
\dot y\quad
\eta_s
\end{bmatrix}^{T}$ denote the transverse coordinates, where $y$ corresponds to the virtual holonomic output and $\eta_s$ is the slip output, which measures deviation from the prescribed slip law. Then
\[
\mathcal Z_{\mathrm{int}}
=
\{(q,\dot q)\in TQ \mid y=0,\ \dot y=0,\ \eta_s=0 \}.
\]

Under the feedback law \eqref{eq:combined_feedback}, the transverse dynamics satisfy \eqref{eq:closed_loop_eta_s}--\eqref{eq:closed_loop_y}. In terms of the state vector $\eta$, they can be written as $\dot\eta = A_\perp \eta$, with
\[
A_\perp =
\begin{bmatrix}
0 & I & 0 \\
-K_p & -K_d & 0 \\
0 & 0 & -k_s
\end{bmatrix}.
\]

Since $k_s>0$ and the second-order system $\ddot y+K_d\dot y+K_p y=0$ is exponentially stable, the matrix $A_\perp$ is Hurwitz. Therefore, the manifold $\mathcal Z_{\mathrm{int}}$ is locally exponentially attractive under the closed-loop continuous-time dynamics.

Let $\Sigma := \mathcal Z_{\mathrm{int}} \cap S$ denote a transversal Poincaré section. The associated Poincaré return map $P:\Sigma \to \Sigma$ assigns, whenever it is well defined, to each point in $\Sigma$ the state obtained after one complete step of the hybrid system. A fixed point $x^\ast \in \Sigma$ satisfying $P(x^\ast)=x^\ast$ corresponds to a periodic walking gait.

The stability of a walking gait is determined by the linearization of the Poincaré map at $x^\ast$, denoted by $A_P := DP(x^\ast)$. If all eigenvalues of $A_P$ lie strictly inside the unit disk, then the fixed point $x^\ast$ is locally exponentially stable, and the associated periodic orbit is locally exponentially orbitally stable.

\begin{theorem}
Assume that the manifold $\mathcal Z_{\mathrm{int}}$ is invariant under the continuous-time dynamics and satisfies $\Delta(\mathcal Z_{\mathrm{int}}\cap S)\subset \mathcal Z_{\mathrm{int}}$. Assume moreover that the reduced Poincaré map on $\Sigma=\mathcal Z_{\mathrm{int}}\cap S$ is well defined and admits a fixed point $x^\ast\in\Sigma$.
Then the full hybrid locomotion system admits a periodic walking gait passing through $x^\ast$.
\end{theorem}

\begin{proof}
Since $\mathcal Z_{\mathrm{int}}$ is invariant under the continuous-time dynamics and is preserved by the reset map, the hybrid dynamics restricted to $\mathcal Z_{\mathrm{int}}$ is well defined. Hence the reduced Poincaré map on $\Sigma=\mathcal Z_{\mathrm{int}}\cap S$ is well posed whenever the corresponding return trajectory exists.

If $x^\ast\in\Sigma$ is a fixed point of the reduced Poincaré map, then the reduced hybrid trajectory starting from $x^\ast$ returns to $x^\ast$ after one step. So, it defines a periodic orbit of the reduced hybrid dynamics on $\mathcal Z_{\mathrm{int}}$. Since $\mathcal Z_{\mathrm{int}}$ is an invariant hybrid manifold of the full system, this orbit is also a periodic orbit of the full hybrid system. Hence the full system admits a periodic walking gait passing through $x^\ast$.
\end{proof}

The following result is a direct consequence of standard Poincaré-map arguments for hybrid systems and restricted hybrid zero dynamics; see, e.g., \cite{morris2005restricted, Westervelt2007}.

\begin{theorem}
Consider the closed-loop hybrid locomotion system with feedback law \eqref{eq:combined_feedback}. Suppose that
\begin{enumerate}
\item the manifold $\mathcal Z_{\mathrm{int}}$ is invariant under the continuous-time dynamics and satisfies $\Delta(\mathcal Z_{\mathrm{int}}\cap S)\subset \mathcal Z_{\mathrm{int}}$,
\item the restricted Poincaré map $P:\Sigma\to\Sigma$, with $\Sigma=\mathcal Z_{\mathrm{int}}\cap S$, is well defined and admits a fixed point $x^\ast$,
\item all eigenvalues of the Jacobian $A_P:=DP(x^\ast)$ lie strictly inside the unit disk.
\end{enumerate}
Then the periodic walking gait associated with $x^\ast$ is locally exponentially orbitally stable.
\end{theorem}

\begin{proof}
By the invariance of $\mathcal Z_{\mathrm{int}}$ under the continuous-time dynamics and the reset condition $\Delta(\mathcal Z_{\mathrm{int}}\cap S)\subset \mathcal Z_{\mathrm{int}}$, the hybrid dynamics restricted to $\mathcal Z_{\mathrm{int}}$ are well defined. Thus the restricted Poincaré map on $\Sigma=\mathcal Z_{\mathrm{int}}\cap S$ is well posed whenever the corresponding return trajectory exists.


By Assumption~2, $x^\ast$ is a fixed point of the restricted Poincaré map. By Assumption~3, all eigenvalues of $A_P=DP(x^\ast)$ lie strictly inside the unit disk. Therefore, $x^\ast$ is a locally exponentially stable fixed point of the restricted return map. Since $\Sigma\subset \mathcal Z_{\mathrm{int}}$ and the transverse closed-loop dynamics converge exponentially to $\mathcal Z_{\mathrm{int}}$, standard results on restricted Poincaré maps and hybrid zero dynamics imply that the corresponding periodic orbit of the full hybrid locomotion system is locally exponentially orbitally stable; see, e.g., \cite{morris2005restricted, Westervelt2007}.
\end{proof}

\section{Numerical Results}

For the numerical study, we consider a structured planar 7-DOF hybrid biped model inspired by the slip-walking framework of Chen, Trkov, and Yi~\cite{chen2017hybrid,TrkovChenYi2019}. The model is chosen so as to remain consistent with the control architecture developed in Sections~II--III: one coordinate describes the longitudinal progression of the gait, while the remaining six coordinates are regulated through virtual holonomic outputs parameterized by B\'ezier polynomials. 

In the numerical implementation, the scalar virtual  nonholonomic constraint is specified by $\eta_s(q,\dot q)=A(q)\dot q-b(q)$, with 
$A(q)\equiv
\begin{bmatrix}
1 & 0 & 0.10 & 0.08 & 0.04 & -0.05 & -0.03
\end{bmatrix}$, while the prescribed slip law is taken as $$b(q)\equiv b_k=v_{\mathrm{nom}}+s_k$$ during the $k$-th stance phase, where $s_k$ denotes the imposed step-dependent slip level. Thus, in the simulations, $A(q)$ is chosen constant and $b(q)$ is implemented as a piecewise-constant reference from step to step. 

Although the simulation model is a structured surrogate rather than a full rigid-body Euler--Lagrange realization, it preserves the dimensional and geometric structure required by the proposed control law: six holonomic outputs, one scalar slip output, and a square stacked decoupling matrix. Accordingly, the numerical study is intended to validate the control
architecture and its regularity properties, rather than to provide a
high-fidelity validation of frictional contact mechanics.

To assess consistency with the theoretical assumptions, we numerically evaluate along the simulated trajectories the regularity conditions used in the control design. In particular, we monitor the scalar slip decoupling term $M_s(q)$ and the smallest singular value of the stacked matrix
$\mathcal A(q,\dot q)=
\begin{bmatrix}
M_s(q)\\
M_h(q,\dot q)
\end{bmatrix}$. Table~\ref{tab:assumptions} summarizes the main indicators. The quantities $\min \sigma_{\min}(\mathcal A)$ and $\min \|M_s\|$ act as structural regularity measures for the simulated model rather than direct performance metrics. In the present realization, they depend only weakly on the closed-loop evolution, which explains why their values are identical in the open-loop and controlled cases. By contrast, the number of successful steps, the terminal pre-impact speed, and the mean slip-output magnitude clearly distinguish the two responses. In particular, under the proposed feedback law the slip output remains close to zero throughout the simulation, while both $M_s(q)$ and $\mathcal A(q,\dot q)$ remain uniformly away from singularity. These observations support, at least numerically for the simulated gait, the assumptions underlying Propositions~1--2 and the stability analysis of Section~III-C.

\begin{table}[h!]
\centering
\small
\setlength{\tabcolsep}{4pt}
\caption{Numerical verification of regularity and performance indicators along the simulated gait.}
\label{tab:assumptions}
\begin{tabular}{lcc}
\hline
Quantity & Open loop & Controlled \\
\hline
Successful steps & 29 & 50 \\
Terminal pre-impact speed $v^\ast$ & 0.2565 & 0.5964 \\
Mean $|\eta_s|$ & 0.6505 & 0.0041 \\
$\min \sigma_{\min}(\mathcal A)$ & 0.4417 & 0.4417 \\
$\min \|M_s\|$ & 1.0106 & 1.0106 \\
\hline
\end{tabular}
\end{table}

To make the numerical procedure explicit, the simulation loop underlying the experiments is summarized as follows.

\begin{algorithm}[h!]
\caption{Hybrid simulation with virtual holonomic and nonholonomic constraints}
\label{alg:simulation}
\begin{algorithmic}[1]
\State Choose a step horizon $N\in\mathbb{N}$ and a step-dependent slip sequence $\{s_k\}_{k=1}^N$.
\State Initialize the state $x_0=(q_0,\dot q_0)$ near the nominal gait.
\For{$k=1,\dots,N$}
    \State Define the stepwise slip reference $b_k$ from the prescribed slip level $s_k$.
    \State Construct the virtual holonomic outputs
    $
    y(q)=h(q)-h_d(\theta(q)),
    $
    where $h_d$ is parameterized by B\'ezier polynomials.
    \State Define the slip output
    $
    \eta_s(q,\dot q)=A(q)\dot q-b_k.
    $
    \State Compute the continuous-time feedback law enforcing both the holonomic and nonholonomic constraints.
    \State Integrate the closed-loop continuous-time dynamics until the impact event
    $
    h_{\mathrm{sw}}(q)=0,\ \dot h_{\mathrm{sw}}(q,\dot q)<0.
    $
    \State Apply the reset map and the leg-swap operation to obtain the next initial condition.
    \State Record the pre-impact velocity, the slip-output history, and the regularity indicators associated with
    $
    M_s(q)
    $
    and the stacked decoupling matrix
    $
    \mathcal A(q,\dot q).
    $
\EndFor
\end{algorithmic}
\end{algorithm}

Algorithm~\ref{alg:simulation} describes the general computational structure used in the numerical study. 
The main experiment uses a 50-step slip sequence with piecewise-constant levels $s_k \in \{0,\;0.015,\;0.03,\;0,\;0.02\}$, applied over consecutive step intervals to emulate alternating ground patches of different tangential slip amounts. 

In the open-loop case, the virtual holonomic outputs are still tracked by the nominal continuous gait controller, but the scalar slip-regulation channel is not activated. In this case, the gait eventually loses stability after 29 successful steps, whereas under the proposed feedback law the walker completes 50 steps without failure. Figure~\ref{fig:variableslip} shows that, under variable slip conditions, the open-loop gait progressively loses forward speed and ultimately fails, while the controlled gait reorganizes after each change in the prescribed slip level and remains persistent over the full simulation horizon. At the same time, the controlled case maintains a very small slip-output error, in agreement with the imposed constraint manifold.

Figure~\ref{fig:hippath} shows the hip trajectory in the sagittal plane for the open-loop and controlled cases. The controlled gait preserves a coherent geometric walking pattern throughout the simulation, in agreement with the step-to-step stabilization shown in Fig.~\ref{fig:variableslip}. Overall, the numerical results indicate that the proposed slip-compatible control design can simultaneously regulate the slip output and preserve robust periodic walking under nonuniform contact conditions.

\begin{figure}[h!]
    \centering
    \includegraphics[width=\columnwidth]{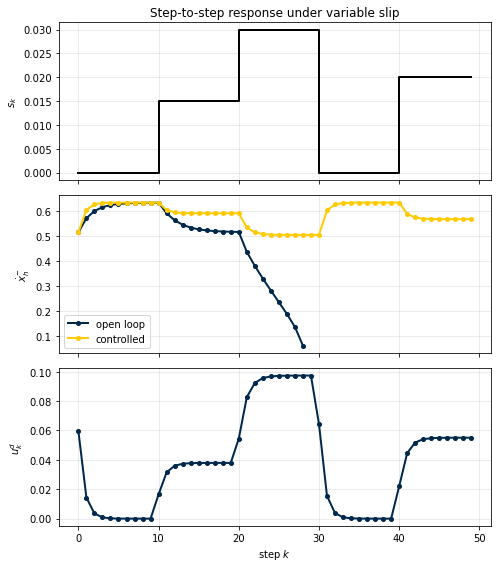}
    \caption{Step-to-step response under variable slip conditions. Top: prescribed slip level at each step. Middle: pre-impact forward velocity for the open-loop and controlled cases. Bottom: slip-output magnitude under the controlled case.}
    \label{fig:variableslip}
\end{figure}

\begin{figure}[h!]
    \centering
    \includegraphics[width=0.9\columnwidth]{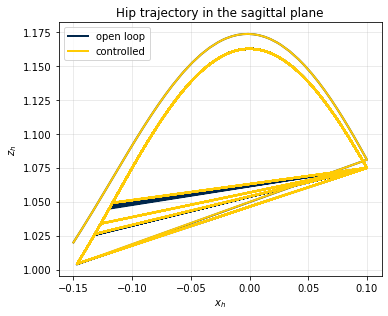}
    \caption{Hip trajectory in the sagittal plane for the open-loop and controlled cases under variable slip conditions.}
    \label{fig:hippath}
\end{figure}

A video illustrating the controlled walking motion under the same variable-slip scenario is provided in \textcolor{blue}{\url{https://youtu.be/EWsuVqacxpU}} The video shows the persistence of the gait over the full 50-step horizon, the adaptation of the walker to the piecewise-constant slip sequence, and the maintenance of a coherent sagittal-plane motion while the slip output remains small throughout the simulation. This evidence complements Figs.~\ref{fig:variableslip}--\ref{fig:hippath} and further supports the effectiveness of the proposed slip-compatible control design.

\section{Conclusions and Future Work}

We developed a control framework for bipedal walking with foot slip based on virtual nonholonomic constraints. The proposed design regulates tangential stance-foot motion while remaining compatible with the virtual holonomic constraints used for gait generation, yielding a slip-compatible hybrid zero dynamics manifold and an orbital stability characterization via the reduced Poincar\'e map. Numerical simulations showed stable walking under variable slip conditions together with effective regulation of the slip output. Future work includes extending the framework to full
Euler-Lagrange models with explicit contact and friction dynamics,
including stick-slip transitions and contact-force feasibility, and
studying more general or terrain-adaptive slip laws.



\bigskip

{\bf Acknowledgment:} We would like to thank Jessy Grizzle for introducing us to the problem of walking with foot slip and for his inspiring contributions to robotic locomotion. 

\bibliographystyle{IEEEtran}
\bibliography{biblio}

\end{document}